\documentclass[11pt]{article}
\usepackage{graphicx, verbatim, fullpage}
\usepackage{hyperref}
\usepackage{amssymb}
\usepackage{amsmath}
\usepackage{amsthm}
\usepackage{enumerate}
\usepackage{multicol}

\def\final{0}  
\def\iflong{\iffalse}
\ifnum\final=0  
\newcommand{\knote}[1]{[{\tiny karthik: \bf #1}]\marginpar{*}}
\newcommand{\snote}[1]{[{\tiny Santosh: \bf #1}]\marginpar{*}}
\newcommand{\sidecomment}[1]{}
\else 
\newcommand{\knote}[1]{}
\newcommand{\snote}[1]{}
\newcommand{\sidecomment}[1]{}
\fi 

\newtheorem{theorem}{Theorem}
\newtheorem{lemma}{Lemma}
\newtheorem{prop}[lemma]{Proposition}

\def\setm{[m]}
\def\setn{[n]}
\def\radius{R}

\def\reals{\mathbb{R}}
\def\integers{\mathbb{Z}}

\def\R{\reals}

\newcommand{\restvector}[2]{{#1}^{_{#2}}}
\newcommand{\restrictedvector}[2]{{#1}^{#2}}
\newcommand{\vproduct}[2]{{#1} {#2}}
\newcommand{\ball}[2]{#2\mathbb{B}_{#1}}
\newcommand{\abs}[1]{\left|#1\right|}
\newcommand{\norm}[1]{\left\|#1\right\|}

\newcommand{\prob}[1]{{\sf Pr}\left(#1\right)}

\newcommand{\vol}[1]{\operatorname{vol}\left(#1\right)}

\newcommand{\var}[1]{\operatorname{Var}\left(#1\right)}
\newcommand{\expectation}[1]{\operatorname{\mathbb{E}}\left(#1\right)}

\title{Integer Feasibility of Random Polytopes}

\author{Karthekeyan Chandrasekaran\footnote{karthe@seas.harvard.edu, School of Engineering and Applied Sciences, 
Harvard University; This work was done while the author was a student at Georgia Institute of Technology supported in part by the Algorithms and Randomness Center (ARC) fellowship and the NSF.}
\and
Santosh Vempala\footnote{vempala@cc.gatech.edu, School of Computer Science, Georgia Institute of Technology; Supported in part by the NSF.}
}

\date{}

\begin{document}

\maketitle

\begin{abstract}
We study integer programming instances over polytopes $P(A,b)=\{x:Ax\le b\}$ where the constraint matrix $A$ is random, i.e., its entries are i.i.d. Gaussian or, more generally, its rows are i.i.d. from a spherically symmetric distribution. The radius of the largest inscribed ball is closely related to the existence of integer points in the polytope. We show that for $m=2^{O(\sqrt{n})}$, there exist constants $c_0 < c_1$ such that with high probability, random polytopes are  integer feasible if the radius of the largest ball contained in the polytope is at least $c_1\sqrt{\log{(m/n)}}$; and integer infeasible if the largest ball contained in the polytope is centered at $(1/2,\ldots,1/2)$ and has radius at most $c_0\sqrt{\log{(m/n)}}$. Thus, random polytopes transition from having no integer points to being integer feasible within a constant factor increase in the radius of the largest inscribed ball. We show integer feasibility via a randomized polynomial-time algorithm for finding an integer point in the polytope. 

Our main tool is a simple new connection between integer feasibility and linear discrepancy.
We extend a recent algorithm for finding low-discrepancy solutions \cite{LM12} to give a constructive upper bound on the linear discrepancy of random matrices. By our connection between discrepancy and integer feasibility, this upper bound on linear discrepancy translates to the radius lower bound that guarantees integer feasibility of random polytopes. 
\end{abstract}

\newpage
\section{Introduction}
Integer Linear Programming (IP) is a general and powerful formulation for combinatorial problems \cite{schrijver-IP-book,nemhauser-wolsey}. One standard variant is the integer feasibility problem: given a polytope $P$ specified by linear constraints ${A}{x}\le b$, find an integer solution in $P$ or report that none exists. The problem is NP-hard and appears in Karp's original list \cite{karp-np-complete}. Dantzig \cite{dantzig60} suggested the possibility of IP being a {\em complete} problem even before the Cook-Levin theory of NP-completeness. The best-known rigorous bound on the complexity of general IP is essentially $n^{O(n)}$ from 1987 \cite{Kannan87}.

While IP in its general form is intractable, several special instances are very interesting and not yet well-understood. One such simple and natural family of instances is randomly generated IP instances, where the constraint matrix $A$ describing random polytopes is drawn from a distribution. 

Random instances have been studied for several combinatorial problems e.g., random-SAT \cite{random-SAT-CF86,random-SAT-CF90,random-SAT-CR92,random-SAT-BFU93,sharp-threshold-graph-SAT-F98}, random knapsack \cite{random-knapsack} and various other graph problems on random graphs \cite{book-random-graphs-bollobas01}. IP is one of the few problems in Karp's original list \cite{karp-np-complete} that has not been satisfactorily understood for random instances. Furst and Kannan \cite{kannan-furst} studied the random single-row subset-sum IP. Their results were generalized to multi-row IP by Pataki et al. \cite{pataki10}. They showed that if each entry in the constraint matrix $A$ is chosen independently and uniformly at random from the discrete set $\{1,2,\ldots,M\}$, then with high probability, a certain reformulation of such random IP instances can be solved efficiently by the branch-and-bound algorithm provided that $M$ is sufficiently large. Their requirement that $M$ be larger than the length of the RHS vector $b$ of the input polytope hardly resembles the flavor of results in well-studied random graph and random satisfiability instances. \\


\noindent{\bf Model for random IPs.}
A random IP instance in our model is described by a random constraint matrix $A\in \R^{m\times n}$ and an RHS vector $b$.
Formally, we obtain random IP instances by generating random polytopes $P(n,m,x_0,R)=\{x\in\R^n:A_ix\le b_i \ \forall\ i\in[m]\}$ as follows: pick a random $m \times n$ matrix $A$ with i.i.d. entries from the Gaussian distribution $N(0,1)$; and a vector $b$ such that the hyperplane corresponding to each constraint is at distance $R$ from $x_0$, i.e., $b_i=R\|A_i\|+A_ix_0$, where $A_i$ is the $i$'th row of $A$.

\begin{figure}[ht]\label{fig:model}
\centering
\includegraphics[scale=0.6]{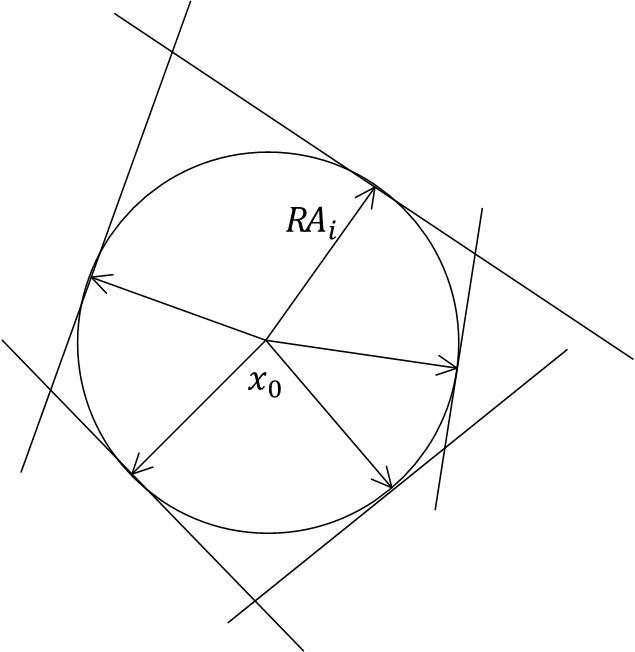}
\caption{Random IP instances $P(n,m,x_0,R)$: random unit vectors $A_i$ describe the normals to the facets and each facet is at distance $R$ from $x_0$.}
\end{figure}

An equivalent geometric interpretation for our model of random polytopes is the following (see Figure \ref{fig:model}): we recall that if each row of the constraint matrix $A$ is a unit vector, then they describe the normals to the facets of the polytope $P=\{x:Ax\le b\}$. 
Thus, the random polytopes $P(n,m,x_0,R)$ in our model are obtained using $m$ facets whose normal vectors are independent uniform random unit vectors in $\R^n$ and such that each facet is at distance $R$ from the point $x_0$.

The condition that all facets are at distance $\radius$ from $x_0$ is equivalent to the condition that $P(n,m,x_0,R)$ contains a ball of radius $\radius$ centered at $x_0$. We study the integer feasibility of $P(n,m,x_0,R)$ for every $x_0$ as a function of the radius $\radius$. As the radius $\radius$ used to generate random polytopes increases, it is likely that the polytopes contain an integer point. \\

\noindent{\bf Contributions.} We show a phase-transition phenomenon regarding integer feasibility of random polytopes with respect to the radius used to generate these polytopes---we show an upper bound on $\radius$ needed to guarantee integer feasibility with high probability for every $x_0\in \R^n$; we show a lower bound that guarantees integer infeasibility with high probability for a fixed $x_0=(1/2,\ldots,1/2)$; our upper and lower bounds differ by a constant factor when $m=2^{O(\sqrt{n})}$. 
We show our upper bound by giving an efficient algorithm to find an integer feasible solution in the feasibility regime. This is an application of a recent constructive proof of discrepancy \cite{LM12}.

Alternatively, our results can be reinterpreted to bear resemblance to the well-known random SAT threshold: consider random polytopes in $n$-dimensions obtained by picking $m$ random tangential hyperplanes to a ball of ``constant'' radius centered at $x_0$. If $m\le c_0n$, then random polytopes are integer feasible for every $x_0$ with high probability and if $m\ge c_1n$, then random polytopes are integer infeasible for $x_0=(1/2,\ldots,1/2)$. Thus, integer feasibility of random polytopes exhibits a phase transition-like behavior when the number of hyperplanes increases beyond a constant times the number of variables, very similar to the satisfiability of random $k$-SAT.

Our main conceptual contribution is a new sufficient condition to guarantee integer feasibility of {\em arbitrary} polytopes. The approach is centered on the simple idea that a polytope is likely to contain an integer point if it contains a large ball. In fact, any polytope in $n$-dimensional space that contains a Euclidean ball of radius at least $\sqrt{n}/2$ is integer feasible. We refine this radius $r(A)$ of the largest inscribed ball that guarantees integer feasibility as a function of the constraint matrix $A$ describing the polytope. This refined radius function is helpful in deriving bounds on the radius of the largest inscribed ball that guarantees integer feasibility of random polytopes. 

With $\radius= \Omega(\sqrt{\log m})$ and $x_0=(1/2,...1/2)$, there is a trivial algorithm --- pick a random $0/1$ vector. Most such vectors will be feasible in $P(n,m,x_0,R)$. But with smaller $\radius$, and arbitrary centers $x_0$, only an exponentially small fraction of nearby integer vectors might be feasible, so such direct sampling/enumeration would not give a feasible integer point. We employ a more careful sampling technique for smaller $R$. This is a simple extension of a recent algorithm for finding low discrepancy solutions \cite{LM12}.

\subsection{Results}
Our main theorem is stated as follows.
\begin{theorem}\label{thm:threshold-radius}
Let $1000n\le m\le 2^{n}$ and 
\begin{align*}
R_0&=\sqrt{\frac{1}{6}\log{\frac{m}{n}}},\ R_1=960\left(\sqrt{\log{\frac{m}{n}}} + \sqrt{\frac{\log{m}\log{(mn)}\log{(m/\log{m})}}{n}}\right).
\end{align*}
Then, 
\begin{enumerate}
\item there exists a randomized polynomial time algorithm that with probability at least $1-(4/m^3)-2me^{-n/96}$ finds an integer point in the random polytope $P(n,m,x_0,R)$ for every $x_0\in \reals^n$ when $R\ge R_1$, and
\item with probability at least $1-2^{-n}-2me^{-n/96}$, the random polytope \\ $P(n,m,x_0=(1/2,\ldots,1/2),R)$ does not contain an integer point when $R\le R_0$.
\end{enumerate}
\end{theorem}

The above results are also applicable in the equivalent random polytope model $P(n,m,x_0,R)$ obtained using random matrices $A$ whose rows are chosen i.i.d. from any spherically symmetric distribution.
\vspace{2mm}

\noindent {Remarks.}
\begin{enumerate}
\item For $m = 2^{O(\sqrt{n})}$, the second term in $R_1$ is of the same order as the first and so the two thresholds are within a constant factor of each other. Thus, in this case, the transition between infeasibility and feasibility happens within a constant factor increase in the radius. 
\item When $m=cn$ for some sufficiently large constant $c$, our theorem shows that a constant radius ball inscribed in random polytopes is sufficient to guarantee integer feasibility with high probability (as opposed to the $\sqrt{n}/2$ radius ball needed in the case of arbitrary polytopes).
\end{enumerate}

Underlying the above theorem is our main conceptual contribution -- a simple yet powerful connection between the radius of the largest inscribed ball that guarantees integer feasibility and the linear discrepancy of the constraint matrix \cite{matousek98,spencer85,spencer-lecture}. We show that if the radius is at least the linear discrepancy of the normalized constraint matrix (each row is normalized to a unit vector), then the polytope contains an integer point. 

The linear discrepancy of a matrix $A\in \reals^{m\times n}$ is defined as follows:
\[
\text{lin-disc}(A):=\max_{x_0\in [0,1]^n}\min_{x\in \{0,1\}^n}\norm{A(x-x_0)}_{\infty}.
\]

\begin{prop}\label{prop:linear-disc-IP}
Every polytope $P_{x_0}(A)=\{x\in \reals^n|\abs{{A_i}{(x-x_0)}}\leq b_i \text{ for }i\in \setm\}$ where $b_i\geq \text{lin-disc}(A)$ contains an integer point for every $x_0\in \reals^n$. 
\end{prop}
We elaborate on Proposition \ref{prop:linear-disc-IP} in Section \ref{sec:discrepancy-connection}. To apply this connection to random IPs, we bound the linear discrepancy of Gaussian matrices.

\begin{theorem}\label{thm:random-hyperplanes-arbitrary-center-contain-integer-point}
Let $A \in \R^{m \times n}$ be a random matrix with i.i.d. entries from $N(0,\sigma^2)$, where $2n\le m\le 2^n$. There exists an algorithm that takes a point $x_0\in \reals^n$ as input and outputs a point $x\in \integers^{n}$ by rounding each coordinate of $x_0$ either up or down such that, for every $i\in [m]$,
\[
\left|\vproduct{A_{i}}{(x-x_0)}\right|
\leq 480\sigma\left(\sqrt{n\log{\frac{m}{n}}}+\sqrt{\log{m}\log{(mn)}\log{\frac{m}{\log{m}}}}\right).
\]
with probability at least $1-(4/m^3)$. Moreover, the algorithm runs in expected time that is polynomial in $n$ and $m$.
\end{theorem}

In terms of classical discrepancy theory, Theorem \ref{thm:random-hyperplanes-arbitrary-center-contain-integer-point} is equivalent to a bound of $O\left(\sigma R_1\sqrt{n}\right)$ on the linear discrepancy of random Gaussian matrices. The integer feasibility in Theorem \ref{thm:threshold-radius} (part 1) follows from Theorem \ref{thm:random-hyperplanes-arbitrary-center-contain-integer-point} by choosing $\sigma^2 = 1$ and observing that with probability at least $1-2me^{-n/96}$, all $m$ random Gaussian vectors in $n$-dimension have length at most $\sqrt{n}$ up to a constant scaling factor.


\subsection{The Discrepancy Connection}\label{sec:discrepancy-connection}
To understand this connection, we begin with a simpler problem -- suppose we seek $-1/1$ points in the polytope (as opposed to integer points). Given a matrix $A\in \reals^{m\times n}$, and a real positive value $r$, consider the polytope $P(A,r)=\{x\in \reals^n: |A_ix|\le r\ \forall\ i\in [m] \}$. The discrepancy of a matrix $A$ is defined to be the least $r$ so that the polytope $P(A,r)$ contains a $-1/1$ point. This is equivalent to the classical definition of discrepancy \cite{matousek98,spencer85,spencer-lecture}:
\[
\text{disc}(A):=\min_{x\in \{-1,+1\}^{n}}\norm{Ax}_{\infty}.
\]
The following proposition is an immediate consequence of this definition. 
\begin{prop}\label{prop:discrepancy-IP}
The polytope $P(A,\text{disc}(A))=\{x\in \reals^n:|A_ix|\le \text{disc}(A)\ \forall\ i\in [m]\}$ contains a $-1/1$ point.
\end{prop}
This is because, the point $x\in \{-1,+1\}^n$ that minimizes discrepancy is in fact contained in the polytope $P(A,\text{disc}(A))$. Thus, if we can evaluate the discrepancy of the constraint matrix $A$, then by verifying whether the infinity norm of the RHS vector is at least disc$(A)$, we have an easy heuristic to verify if the polytope contains a $-1/1$ point. Hence, if each row of $A$ is a normalized unit vector, then the polytope $Ax\le b$ contains a $-1/1$ point if it contains a ball of radius at least disc$(A)$ centered at the origin.

The related notion of linear discrepancy helps in providing a sufficient condition for {\em integer feasibility} (as opposed to $-1/1$ feasibility) of arbitrary polytopes. Proposition \ref{prop:linear-disc-IP}, similar to Proposition \ref{prop:discrepancy-IP}, is an immediate consequence of the definition of linear discrepancy. 
This is because, by linear transformation, we may assume that $x_0$ is in the fundamental cube defined by the standard basis unit vectors. Thus, if each row of the matrix $A\in \reals^{m\times n}$ is a unit vector, then the linear discrepancy of the constraint matrix gives one possible radius of the largest inscribed ball that guarantees integer feasibility of polytopes described by the constraint matrix $A$.

We observe that the heuristic approach that follows from Proposition \ref{prop:linear-disc-IP} to verify integer feasibility of arbitrary polytopes requires the computation of linear discrepancy of arbitrary matrices. The seemingly related problem of computing the discrepancy of arbitrary matrices even to within an approximation factor of $\sqrt{n}$ is known to be NP-hard \cite{charikar11}. In work subsequent to our result, Nikolov, Talwar and Zhang \cite{geometry-of-DP-HD-NTZ13} have shown that \emph{hereditary discrepancy}, which is an upper bound on linear discrepancy (see Theorem \ref{thm:lindisc-less-than-herdisc} below), can be efficiently computed to within an approximation factor of $poly(\log{m},\log{n})$; this could potentially be useful as a heuristic to verify integer feasibility. 

In order to understand the integer feasibility of random polytopes using this approach, we seek a bound on the linear discrepancy of random matrices that holds with high probability. We obtain such a tight bound for random matrices algorithmically by extending a recent constructive algorithm to minimize discrepancy \cite{LM12} to an algorithm to minimize {\em linear} discrepancy and using appropriate concentration inequalities for random matrices. Our infeasibility threshold is also based on discrepancy---we begin with a lower bound on the discrepancy of random matrices, which excludes any $0/1$ point from being a feasible solution for $P(n,m,x_0=(1/2,\ldots,1/2),R_0)$, and then extend this to exclude all integer points.

\section{Preliminaries}
\subsection{Related Work}
The central quantity that leads to all known bounds on discrepancy and linear discrepancy in the literature is hereditary discrepancy defined as follows:
\[\text{herdisc}(A):=\max_{S\subseteq \setn}\text{disc}(A^S)\]
where $A^S$ denotes the submatrix of $A$ containing columns indexed by the set $S$. For a matrix $A \in \R^{m\times n}$ and any $S\subseteq\setn$, let $\restvector{A_i}{S}$ denote the $i$'th row vector $A_i$ restricted to the coordinates in $S$. The best known bound on discrepancy of arbitrary matrices is due to Spencer \cite{spencer85}.
\begin{theorem}\label{theorem:discrepancy-bound-with-max-entry} \cite{spencer85}
For any matrix $A\in \reals^{m\times n}$, any subset $S\subseteq[n]$, there exists a point $z\in \{-1,+1\}^{|S|}$ such that
\[
\abs{\vproduct{\restrictedvector{A_i}{S}}{z}}\leq 11\sqrt{|S|\log{\frac{2m}{|S|}}}\max_{i\in\setm,j\in S}|A_{ij}|\ \forall\ i\in[m].
\]
\end{theorem}

Lov\'{a}sz, Spencer and Vesztergombi \cite{lovasz86} showed the following relation between hereditary discrepancy and linear discrepancy.
\begin{theorem}\label{thm:lindisc-less-than-herdisc} \cite{lovasz86}
For any matrix $A$, $\text{lindisc}(A)\leq \text{herdisc}(A)$.
\end{theorem}


\subsection{Concentration Inequalities}
We will use the following standard tail bounds.

\begin{lemma}\label{lemma:gaussian-anti-concentration-bound}
Let $Y$ be a random variable drawn from the Gaussian distribution $N(0,\sigma^2)$. For any $\lambda >0$,
\[
\prob{Y\le \lambda \sigma} \le
\min{\left\{1-{\sqrt{\frac{1}{2\pi}}\left(\frac{\lambda}{\lambda^2+1}\right)e^{-\frac{\lambda^2}{2}}},\lambda\sqrt{\frac{1}{2\pi}}\right\}}.
\]
\end{lemma}

\begin{lemma}\label{lemma:gaussian-tail-bound}
Let $Y$ be a random variable drawn from the Gaussian distribution $N(0,\sigma^2)$. For any $\lambda\geq1$,
\[
\prob{|X|\geq \lambda\sigma}\leq 2e^{-\frac{\lambda^2}{2}}.
\]
\end{lemma}

\begin{lemma}\label{lemma:bounded-norm-gaussian-variables}\cite{dasgupta}
Let $X_1,\ldots, X_{n}$ be independent random variables each drawn from the Gaussian distribution $N(0,\sigma^2)$. For any $\lambda>0$,
\[
\prob{|\sum_{j\in [n]}X_{j}^2-n\sigma^2|\geq \lambda\sqrt{n}\sigma^2}\leq 2e^{-\frac{\lambda^2}{24}}.
\]
\end{lemma}

\begin{lemma}\label{lemma:bounded-dot-product}\cite{matousek96-entropyfunction}
Let $X_1,\ldots, X_{n}$ be independent random variables each drawn uniformly from $\{-1,+1\}$. For a fixed set of vectors $a_1,\ldots,a_m\in \reals^n$, a fixed subset $S\subseteq \setn$, and any $\lambda \ge 0$,
\[
\prob{|\sum_{j\in S}a_{ij}X_j|\geq \lambda}\leq 2e^{-\frac{\lambda^2}{2\sum_{j\in S}a_{ij}^2}}.
\]
\end{lemma}

\section{Linear discrepancy of random matrices}
Our first step towards an algorithm to identify an integer point in random polytopes is an algorithm to find small linear discrepancy solutions for random Gaussian matrices. The main goal of this section is to prove the bound on linear discrepancy of Gaussian matrices (Theorem \ref{thm:random-hyperplanes-arbitrary-center-contain-integer-point}).\\


\noindent {\bf Implications of known bounds.}
It is tempting to use known concentration inequalities in conjunction with Spencer's result (Theorem \ref{theorem:discrepancy-bound-with-max-entry}) to bound the hereditary discrepancy of Gaussian matrices; this would in turn lead to a bound on the linear discrepancy of Gaussian matrices by Theorem \ref{thm:lindisc-less-than-herdisc}. In this setting, each entry $A_{ij}$ is from $N(0,\sigma^2)$. Using standard concentration for $|A_{ij}|$ and a union bound to bound the maximum entry $|A_{ij}|$ leads to the following weak bound: with high probability, the polytope $P=\{x\in \reals^n|\abs{{A_i}{(x-x_0)}}\leq b_i\ \text{for }i\in \setm\}$ with $b_i=\Omega(\sigma\sqrt{n\log{mn}\log{(2m/n)}})$ contains an integer point for any $x_0\in \reals^n$. This is too weak for our purpose (recall that $\sqrt{n}$ radius ball in arbitrary polytopes already guarantees integer feasibility and our goal is to guarantee integer feasibility with smaller inscribed ball in {\em random} polytopes).\\

\noindent {\bf Our Strategy.} Our overall strategy to bound discrepancy is similar to that of Spencer's: As a first step, show a partial coloring with low discrepancy -- i.e., for any subset $U\subseteq [n]$, there exists a point $z\in \{0,-1,+1\}^{|U|}$ with at least $|U|/2$ non-zero coordinates such that $|\vproduct{A_i^U}{z}|$ is small. Next for any $S\subseteq [n]$, repeatedly use the existence of this partial vector to derive a vector $x\in \{-1,1\}^{|S|}$ with small discrepancy -- start with $x=0$, $U=S$ and use $z$ to fix at least half of the coordinates of $x$ to $+1$ or $-1$; then take $U$ to be the set of coordinates that are set to zero in the current $x$ and use $z$ to fix at least half of the remaining coordinates of $x$ to $+1$ or $-1$; repeat this until all coordinates of $x$ are non-zero. Since at most $|U|/2$ coordinates are set to zero in each round of fixing coordinates, this might repeat at most $\log{|S|}\le \log{n}$ times. The total discrepancy is bounded by the sum of the discrepancies incurred in each round of fixing. Thus, the goal is to bound the discrepancy incurred in each partial coloring round.

The discrepancy incurred for the $i$'th constraint by the partial coloring can be bounded as follows\footnote{This is an improvement on the bound shown by Spencer: $\abs{\vproduct{\restvector{A_i}{S}}{z}}=O\left(\max_{i\in [m],j\in S}|A_{ij}|\sqrt{|S|\log{\frac{2m}{|S|}}}\right)$ which can be recovered from (\ref{ineq:3-gen-bound}). The proof of (\ref{ineq:3-gen-bound}) is identical to the proof of Spencer's bound except for a stronger concentration inequality. We avoid the non-constructive proof for simplicity of presentation; we use an alternative algorithmic proof that follows from Lovett-Meka's partial coloring algorithm (see Lemma \ref{lemma:algorithm-substantial-partial-coloring}).}:
\begin{align}
\abs{\vproduct{\restvector{A_i}{U}}{z}}\leq 4\norm{\restvector{A_i}{U}}\sqrt{\log{\frac{2m}{|U|}}}\ \ \forall\ i\in [m],\ U\subseteq \setn. \label{ineq:3-gen-bound}
\end{align}

\noindent {\bf Bounding discrepancy of partial vector.} 
The discrepancy bound 
for the $i$'th constraint given in (\ref{ineq:3-gen-bound}) depends on the length of the vector $\restvector{A_{i}}{U}$. We describe a straightforward approach that does not lead to tight bounds.\\

\noindent {\bf Approach 1.} It is straightforward to obtain $\norm{\restvector{A_i}{U}}\leq 2\sigma\sqrt{|U|\log{mn}}$ with high probability for random Gaussian vectors $A_i$ using well-known upper bound on the maximum coefficient of $\restvector{A_i}{U}$. This leads to an upper bound of 
\[
8\sigma\sqrt{|S|\log{(mn)}\log{\frac{2m}{|S|}}}
\]
on the discrepancy of $\restvector{A}{S}$. Although this bound on the discrepancy of $\restvector{A}{S}$ is good enough when the cardinality of $S$ is smaller than some threshold, it is too large for large sets $S$. E.g., when $S=\setn$, this gives a total discrepancy of at most $O(\sigma\sqrt{n\log{(mn)}\log{(2m/n)}})$.\\

\noindent {\bf New Approach.} In order to obtain tighter bounds, we bound the length of partial vectors $\restvector{A_i}{U}$ when each entry in the vector is from $N(0,\sigma^2)$ (as opposed to bounding the maximum coefficient). Using Lemma \ref{lemma:bounded-norm-gaussian-variables}, we will show that
\[
\norm{\restvector{A_i}{U}}=O\left(\sigma\sqrt{|U|}\left(\log{\frac{en}{|U|}}\right)^{\frac{1}{4}}\right)
\]
for every $U\subseteq [n]$ of size larger than $\log{m}$ with probability at least $1-1/m^5$. Consequently, the total discrepancy incurred while the number of coordinates to be fixed is larger than $\log{m}$ is bounded by a geometric sum which is at most
\[
O\left(\sigma\sqrt{n\log{\frac{m}{n}}}\right).
\]
When the number of coordinates to be fixed is less than $\log{m}$, we use Approach 1 to bound the length of partial vectors, which in turn implies the required bound on the total discrepancy. 

\subsection{Bounding lengths of Gaussian subvectors}
\begin{lemma}\label{lemma:max-norm-concentration}
Let $A\in \reals^{m\times n}$ be a matrix whose entries are drawn i.i.d. from the Gaussian distribution $N(0,\sigma^2)$. Then, 
\[
\prob{\norm{\restvector{A_i}{S}}\leq 2\sigma\sqrt{|S|\log{mn}}\ \forall S\subseteq [n],\ \forall i\in \setm}\ge 1-\frac{1}{(mn)^3}.
\]
\end{lemma}
\begin{proof}
By Lemma \ref{lemma:gaussian-tail-bound}  and union bound over the choices of $i\in [m], j\in [n]$, all entries $|A_{ij}|\le 2\sigma\sqrt{\log{mn}}$ with probability at least $1-1/(mn)^3$. Now, the squared length is at most the squared maximum entry multiplied by the number of coordinates.
\end{proof}

Next we obtain a bound on the length of $\restvector{A_i}{S}$ when $|S|$ is large.
\begin{lemma}\label{lemma:gaussian-subvectors}
Let $A\in \reals^{m\times n}$ be a matrix whose entries are drawn i.i.d. from $N(0,\sigma^2)$ where $m\le 2^{n}$. Then, 
\[
\prob{\exists S\subseteq [n], |S|\ge \log{m}, \exists i\in [m]: 
\norm{\restvector{A_i}{S}}^2\ge 10\sigma^2 |S|\sqrt{\log{\left(\frac{en}{|S|}\right)}+\frac{1}{|S|}\log{m}}
} \le \frac{1}{m^5}.
\]
\end{lemma}
\begin{proof}
Let 
\[
\lambda_s:=10s\sqrt{\log{\left(\frac{en}{s}\right)}+\frac{1}{s}\log{m}}\ .
\]
Fix a subset $S\subseteq [n]$ of size $|S|=s$ and $i\in [m]$. Then, by Lemma \ref{lemma:bounded-norm-gaussian-variables}, we have that
\[
\prob{\norm{\restvector{A_i}{S}}^2\ge \lambda_s\sigma^2}
\le 2e^{-\frac{\lambda_s^2}{24s}}.
\]
Hence, $\prob{\exists S\subseteq [n], |S|=s, \exists i\in [m]: 
\norm{\restvector{A_i}{S}}^2\ge \lambda_s \sigma^2 
}
\le 2e^{-\frac{\lambda_s^2}{24s}}\cdot \binom{n}{s}\cdot m$
\[
\le 2e^{-\frac{\lambda_s^2}{24s}}\cdot \left(\frac{en}{s}\right)^{s}\cdot m
\le 2e^{-\frac{\lambda_s^2}{24s}+s\log{\frac{en}{s}}+\log{m}}
\le 2e^{-3\left(s\log{\frac{en}{s}}+\log{m}\right)}.
\]

Therefore, $\prob{\exists S\subseteq [n], |S|\ge \log{m}, \exists i\in [m]:
\norm{\restvector{A_i}{S}}^2\ge \lambda_{|S|} \sigma^2
}$
\begin{align*}
&=\prob{\exists s\in \{\log{m},\ldots,n\}, \exists S\subseteq [n], |S|=s, \exists i\in [m]:
\norm{\restvector{A_i}{S}}^2\ge \lambda_{|S|} \sigma^2
}\\
&\le \sum_{s=\log{m}}^{n}2e^{-3\left(s\log{\frac{en}{s}}+\log{m}\right)}
= \left(\frac{2}{m^3}\right) \sum_{s=\log{m}}^{n}e^{-3s\log{\frac{en}{s}}}
\le \left(\frac{2n}{m^3}\right) e^{-3\log{m}\log{\frac{en}{\log{m}}}} 
\le \frac{1}{m^5}.
\end{align*}

The last but one inequality is because the largest term in the sum is $e^{-3\log{m}\log{\frac{en}{\log{m}}}}$. The last inequality is because $n\ge \log{m}$.
\end{proof}

\subsection{Algorithmic Linear Discrepancy}
Our algorithm is essentially a variation of Lovett-Meka's algorithm for constructive discrepancy minimization \cite{LM12}. Lovett-Meka \cite{LM12} provide a constructive partial coloring algorithm matching Spencer's bounds. The main difference in their approach from that of Spencer's is that, the partial coloring algorithm outputs a fractional point $z\in [-1,1]^{|U|}$ such that at least $|U|/2$ coordinates are close to being $1$ or $-1$. After at most $\log{|S|}$ rounds, all coordinates are close to being $1$ or $-1$; a final randomized rounding step increases the total discrepancy incurred only by a small amount. 

Their partial coloring algorithm can easily be extended to minimize linear discrepancy as opposed to discrepancy. In each partial coloring round, their algorithm starts with a point $x\in [-1,1]^n$ and performs a random walk to arrive at a vector $y$ such that the discrepancy overhead incurred by $y$ (i.e., $|\vproduct{A_i}{(y-x)}|$) is small. Further, at least half of the coordinates of $x$ that are far from $1$ or $-1$ are close to $1$ or $-1$ in $y$. This can be extended to an algorithm which, in each phase, starts with a point $x\in [0,1]^n$, and performs a random walk to arrive at a vector $y$ such that the discrepancy overhead incurred by $y$ (i.e., $|\vproduct{A_i}{(y-x)}|$) is small. Further, at least half of the coordinates of $x$ that are far from $0$ or $1$ are close to $0$ or $1$ in $y$. The functionality of such a partial coloring algorithm is summarized in the following lemma. In the rest of this section, given $x\in [0,1]^n$, $\delta$, let $B(x):=\{j\in [n]:\delta< x(j)< 1-\delta\}$.

\begin{lemma}\label{lemma:algorithm-substantial-partial-coloring}\cite{LM12}
Given $x\in [0,1]^n$, $\delta\in (0,0.5]$, $A_1,\ldots,A_m\in \reals^n$, $c_1,\ldots,c_m\ge 0$ such that $\sum_{i=1}^m \text{exp}(-c_i^2/16)\le |B(x)|/16$, there exists a randomized algorithm which with probability at least $0.1$ finds a point $y\in [0,1]^n$ such that
\begin{enumerate}
\item $|\vproduct{A_i}{(y-x)}|\le c_i||\restrictedvector{A_i}{B(x)}||_2\ \forall\ i\in [m]$,
\item $|B(y)|\le  |B(x)|/2$ 
\item If $j\in [n]\setminus B(x)$, then $y(j)=x(j)$.
\end{enumerate}
Moreover, the algorithm runs in time $O((m+n)^3\delta^{-3}\log{(nm/\delta)})$. 
\end{lemma}

We denote the algorithm specified in Lemma \ref{lemma:algorithm-substantial-partial-coloring} as Edge-Walk. To minimize the linear discrepancy of random Gaussian matrices, we repeatedly invoke the Edge-Walk algorithm. We repeat each invocation until it succeeds, so our algorithm is a Las Vegas algorithm. Each successful call reduces the number of coordinates that are far from being integer by at least a factor of $1/2$. Thus, we terminate in at most $\log{n}$ successful calls to the algorithm. Further, the total discrepancy overhead incurred by $x$ is at most the sum of the discrepancy overhead incurred in each successful call. The sum of the discrepancy overheads will be bounded using Lemmas \ref{lemma:max-norm-concentration} and \ref{lemma:gaussian-subvectors}. Finally, we do a randomized rounding to obtain integer coordinates from near-integer coordinates. By standard Chernoff bound, the discrepancy incurred due to randomized rounding will be shown to be small. 

\begin{figure}[ht]
\caption{Algorithm Round-IP} \label{alg:round-IP} \vspace{1mm}
\fbox{\parbox{5.7in}{
\begin{minipage}{5.5in}
\noindent Input: Point $x_0\in \reals^n$, matrix $A\in \reals^{m\times
n}$ where each $A_{ij}\sim N(0,\sigma^2).$ \\
\noindent Output: An integer point $z$. 
\begin{enumerate}
\item {\bf Initialize.} $x=x_0-\lfloor x_0\rfloor$, $\delta=1/8\log{m}$, $B(x):=\{j\in [n]:\delta< x(j)< 1-\delta\}$, $c_i=8\sqrt{\log{(m/|B(x)|)}}$ for every $i\in [m]$.
\item While($|B(x)|>0$)
\begin{enumerate}[(i)]
\item {\bf Edge-Walk.} $y\leftarrow$ Edge-Walk($x$, $\delta$, $A_1,\ldots,A_m$, $c_1,\ldots,c_m$).
\item {\bf Verify and repeat.} $B(y):=\{j\in [n]:\delta< y(j)< 1-\delta\}$. \\ 
If $|B(y)|>|B(x)|/2$ or $|\vproduct{A_i}{(y-x)}|>c_i||\restrictedvector{A_i}{B(x)}||_2$ for some $i\in [m]$, then return to (i).
\item {\bf Update.} $x\leftarrow y$, $B(x)=\{j\in [n]:\delta< x(j)< 1-\delta\}$, $c_i=8\sqrt{\log{(m/|B(x)|)}}$ for every $i\in [m]$.
\end{enumerate}
\item {\bf Randomized Rounding.} For each $j\in [n]$ set
\begin{align*}
z(j)=
\begin{cases}
& \lceil x_0(j)\rceil \mbox{ with probability $x(j)$},\\
& \lfloor x_0(j)\rfloor \mbox{ with probability $1-x(j)$}.
\end{cases}
\end{align*}
\item Output $z$.
\end{enumerate}
\end{minipage}}}
\end{figure}

\begin{proof}[Proof of Theorem \ref{thm:random-hyperplanes-arbitrary-center-contain-integer-point}]
Without loss of generality, we may assume that $x_0\in [0,1]^n$ and our objective is to find $x\in \{0,1\}^n$ with low discrepancy overhead. We use Algorithm Round-IP given in Figure \ref{alg:round-IP}. We will show that, with probability at least $1-4/m^3$, it outputs a point $z\in \{0,1\}^n$ such that
\[
|\vproduct{A_i}{(z-x_0)}|\leq 480\sigma
\left(\sqrt{n\log{\frac{m}{n}}}+\sqrt{\log{m}\log{(mn)}\log{\frac{m}{\log{m}}}}\right).
\]

Let $\overline{x}$ denote the vector at the end of Step 2 in Algorithm Round-IP and let $x_k$ denote the vector $x$ in Algorithm Round-IP after $k$ successful calls to the Edge-Walk algorithm. By a successful call, we mean that the call passes the verification procedure 2(ii) without having to return to 2(i). Let $S_k=B(x_k)$. We first observe that after $k-1$ successful calls to the Edge-Walk subroutine, we have $\sum_{i=1}^{m}exp(-c_i^2/16)\le |S_k|/16$ by the choice of $c_i$s. By Lemma \ref{lemma:algorithm-substantial-partial-coloring}, the discrepancy overhead incurred in the $k$'th successful call to the Edge-Walk subroutine is
\begin{align*}
|\vproduct{A_i}{(x_k-x_{k-1})}|
&\leq 8\norm{\restrictedvector{A_i}{S_k}}\sqrt{\log{\frac{m}{|S_k|}}}.
\end{align*}

Consequently, the total discrepancy is bounded by the sum of the discrepancy overhead incurred in each run. The discrepancy overhead incurred in the $k$'th successful run, where $k:|S_k|\ge \log{m}$, is at most
\begin{align*}
8\sigma\sqrt{10|S_k|\log{\frac{m}{|S_k|}}}\left(\log{\frac{en}{|S_k|}}+\frac{1}{|S_k|}\log{m}\right)^{\frac{1}{4}}
&\le 40\sigma\sqrt{|S_k|\log{\frac{m}{|S_k|}}}\left(\log{\frac{en}{|S_k|}}\right)^{\frac{1}{4}}
\end{align*}
with probability at least $1-(1/m^5)$. This is using the bound on the length of $\restrictedvector{A_i}{S_k}$ by Lemma \ref{lemma:gaussian-subvectors}.

Let $k_1$ be the largest integer such that $|S_{k_1}|>\log{m}$. Thus, with probability at least $1-(1/m^5)$, the discrepancy overhead incurred after $k_1$ successful calls to the Edge-Walk subroutine is at most
\[
D_1:=40\sigma\sum_{k=0}^{\log{\frac{n}{\log{m}}}}\sqrt{n2^{-k}\left(\log{\frac{m}{n2^{-k}}}\right)\sqrt{\log{\frac{e}{2^{-k}}}}}\le 240\sigma\sqrt{n\log{\frac{m}{n}}}.
\]
The upper bound on $D_1$ follows from the following inequalities (by setting $A=m/n$),
\begin{align}
\sqrt{\left(\log{\frac{A}{2^{-k}}}\right)\sqrt{\log{\frac{e}{2^{-k}}}}} &\le \left(\sqrt{\log{A}}+\sqrt{k\log{2}}\right) \left(1+k\log{2}\right)^{\frac{1}{4}}\ \forall\ A\ge 1, \label{disc-ineq1}\\
\sum_{k=0}^{\infty}\sqrt{2^{-k}\left(\log{A}\right)}\left(1+k\log{2}\right)^{\frac{1}{4}} &\le 5\sqrt{\log{A}},\label{disc-ineq2}\\
\sum_{k=0}^{\infty}\sqrt{2^{-k}\cdot k\log{2}}\left(1+k\log{2}\right)^{\frac{1}{4}}&\le 2(\log{2})^{3/4}\sum_{k=0}^{\infty}\sqrt{2^{-k}}k^{3/4}\le 10. \label{disc-ineq3}
\end{align}

By Lemma \ref{lemma:algorithm-substantial-partial-coloring}, the discrepancy overhead incurred in the $k$'th successful call to the Edge-Walk subroutine, where $k:|S_k|\le \log{m}$, is
\[
|\vproduct{A_i}{(x_k-x_{k-1})}|
\leq 8\norm{\restrictedvector{A_i}{S_k}}\sqrt{\log{\frac{m}{|S_k|}}} 
\leq 16\sigma\sqrt{n2^{-k}\log{(mn)}\log{\frac{m}{n2^{-k}}}}
\]
with probability at least $1-1/(mn)^3$. Here, the second inequality is by using Lemma \ref{lemma:max-norm-concentration} and $|S_k|\leq n2^{-k}$. Since each successful call to the Edge-Walk subroutine reduces $B(x)$ by at least half, the number of successful Edge-Walk subroutine calls is at most $\log{n}$.

Thus, with probability at least $1-1/(mn)^3$, the discrepancy overhead incurred by Step 2 in successful rounds $k:|S_k|\le \log{m}$ is at most
\[
D_2:=\sum_{k=\log{\frac{n}{\log{m}}}}^{\log{n}} 16\sigma\sqrt{n2^{-k}\log{(mn)}\log{\frac{m}{n2^{-k}}}}
\]
Now, using the inequalities (\ref{disc-ineq1}), (\ref{disc-ineq2}) and (\ref{disc-ineq3}),
\[
D_2\le 32\sigma\sqrt{\log{m}\log{(mn)}\log{\frac{m}{\log{m}}}}.
\]

Hence, with probability at least $(1-1/m^5)(1-1/(mn)^3)$, at the end of Step 2, we obtain a point $\overline{x}$ such that $\overline{x}\in [0,1]^n$ and $\overline{x}(j)\ge 1-\delta$ or $\overline{x}(j)\le \delta$ for every $j\in [n]$ and the total discrepancy overhead is bounded as follows:
\begin{align*}
\max_{i\in[m]}{|\vproduct{A_i}{(\overline{x}-x_0)}|}&\leq
D_1+D_2 \le 240\sigma\left(\sqrt{n\log{\frac{m}{n}}} + \sqrt{\log{m}\log{(mn)}\log{\frac{m}{\log{m}}}}\right).
\end{align*}

Next we show that the randomized rounding performed in Step 3 incurs small discrepancy. Consider a coordinate $j\in [n]$ that is rounded. Then, 
\begin{align*}
\expectation{z(j)-\overline{x}(j)}&=0,\\
\var{z(j)-\overline{x}(j)}&\le \delta,
\end{align*}
and thus,
\[
\Delta_i^2:=\var{\sum_{j=1}^n A_{ij}(z(j)-\overline{x}(j))}\le ||A_{i}||^2\delta.
\]
Therefore, for $i\in [m]$, by Chernoff bound, 
\[
\prob{|\sum_{j=1}^n A_{ij}(z(j)-\overline{x}(j))|\ge 4\Delta_i\sqrt{\log{m}}}\le \frac{2}{m^8}.
\]
Hence, by union bound, we get that $|\vproduct{A_i}{(z-\overline{x})}|\le 4\Delta_i\sqrt{\log{m}}\le 4\norm{A_i}$ for every $i\in [m]$ with probability at least $1-1/m^7$. Now, applying Lemma \ref{lemma:gaussian-subvectors}, we get that $|\vproduct{A_i}{(z-\overline{x})}|\le 4\sigma\sqrt{n}$ with probability at least $(1-1/m^5)(1-1/m^7)\ge 1-2/m^5$. Thus, 
\[
|\vproduct{A_i}{(z-x_0)}|
\le |\vproduct{A_i}{(z-\overline{x})}|+|\vproduct{A_i}{(\overline{x}-x_0)}|
\le 480\sigma \left(\sqrt{n\log{\frac{m}{n}}}+\sqrt{\log{m}\log{(mn)}\log{\frac{m}{\log{m}}}}\right)\ \forall\ i\in[m]
\]
with probability at least $(1-1/m^5)(1-1/(mn)^3)(1-2/m^5)\ge 1-4/m^3$.

Finally, we compute the running time of the algorithm. Each call to the Edge-Walk subroutine succeeds with probability $0.1$. Hence, the expected number of calls to the Edge-Walk subroutine is at most $10\log{n}$. Since each call to the Edge-Walk subroutine takes $O((m+n)^3\log^3{m}\log{(nm\log{m})})$ time, the expected number of calls is $O(\log{n})$ and the number of steps before each call is $O(m+n)$, the total number of steps is at most $O((m+n)^4\log{n}\log^3m\log{(nm\log{m})})$.
\end{proof}

\section{Radius for integer infeasibility}
The upper bound $R_1$ for the radius in Theorem \ref{thm:threshold-radius} will follow from the linear discrepancy bound given in Theorem \ref{thm:random-hyperplanes-arbitrary-center-contain-integer-point}. For the lower bound, we show the following result for Gaussian matrices.

\begin{lemma}\label{lemma:linear-discrepancy-lower-bound}
For $m\ge 1000n$, let $A\in \reals^{m\times n}$ be a matrix whose
entries are chosen i.i.d. from the Gaussian distribution
$N(0,\sigma^2)$. Let $x_0:=(1/2,\ldots,1/2)\in \reals^n$. Then,
\[
\prob{\exists\ x\in \integers^n:A_i(x-x_0)\le
\frac{\sigma}{2}\sqrt{n\log{\frac{m}{n}}}\ \forall i\in [m]}\le \frac{1}{2^n}.
\]
\end{lemma}

We first show a lower bound on the radius necessary for the random polytope $P(n,m,0,R)$ to contain an integer point with all nonzero coordinates. Lemma \ref{lemma:linear-discrepancy-lower-bound} will follow from the choice of $x_0$.

\begin{lemma}\label{lemma:discrepancy-lower-bound1}
For $m\ge 1000n$, let $A\in \reals^{m\times n}$ be a matrix whose entries are chosen i.i.d. from the Gaussian distribution $N(0,\sigma^2)$. 
Then,
\[
\prob{\exists\ x\in \integers^n:A_ix\le
\sigma\sqrt{n\log{\frac{m}{n}}}\ \forall i\in [m]}\le \frac{1}{2^n}.
\]
\end{lemma}
\begin{proof}
For each $r>0$, we define the set
\[
U_{r}:=\integers^n\cap \{x:\norm{x}=r, |x_j|>0\ \forall j\in [n]\}.
\]
We will show that with probability at least $1-2^{-n}$ (over the choices of the matrix $A$), there does not exist $x\in \cup_{r\ge 0}U_r$ satisfying all the $m$ inequalities. We first observe that $U_r$ is non-empty only if $r\ge \sqrt{n}$. Fix  $r\ge \sqrt{n}$ and a point $x\in U_r$. Now, for $i\in [m]$, since each $A_{ij}$ is chosen from $N(0,\sigma^2)$, the dot product $A_i x$ is distributed according to the normal distribution $N(0,r^2\sigma^2)$. Let
\begin{align*}
P_x&:= \prob{A_ix\le \sigma\sqrt{n\log{\frac{m}{n}}}\ \forall\ i\in [m]},\\
P_r&:=\prob{\exists x\in U_r:A_ix\le \sigma\sqrt{n\log{\frac{m}{n}}}\ \forall\ i\in [m]}.
\end{align*}
By union bound,
\[
P_r\le \sum_{x\in U_r} P_x\le |U_r|\max_{x\in U_r}P_x.
\]
We will obtain an upper bound on $P_x$ that depends only on $r$.
To bound the size of the set $U_r$, we observe that every point in $U_r$ is an integer point on the surface of a sphere of radius $r$ centered around the origin and hence is contained in an euclidean ball of radius $r+1$ centered around the origin. Thus, $|U_r|$ can be bounded by the volume of the sphere of radius $r+1\le 2r$ centered around the origin:
\[
|U_r|\le \vol{\ball{0}{2r}}\le \left(2r\sqrt{\frac{2\pi e}{n}}\right)^n \le \left(\frac{10r}{\sqrt{n}}\right)^n.
\]

Next we bound $P_r$. We have two cases.\\
\noindent Case 1. Let $r\in \left[\sqrt{n},\sqrt{n\log{(m/n)}}\right]$. Since $A_ix$ is distributed according to $N(0,r^2	\sigma^2)$, by Lemma \ref{lemma:gaussian-anti-concentration-bound},
\[
\prob{A_ix\le \sigma\sqrt{n\log{\frac{m}{n}}}}\le 1-\frac{1}{\sqrt{2\pi}}\left(\frac{r\sqrt{n\log{\frac{m}{n}}}}{r^2+n\log{\frac{m}{n}}}\right)\cdot\left(\frac{n}{m}\right)^\frac{n}{2r^2}.
\]
Since each $A_{ij}$ is chosen independently, we have that
\begin{align*}
P_x
&=\prod_{i=1}^m \prob{A_ix\le \sigma\sqrt{n\log{\frac{m}{n}}}}\\
&< \left(1-\frac{1}{\sqrt{2\pi}}\left(\frac{r\sqrt{n\log{\frac{m}{n}}}}{r^2+n\log{\frac{m}{n}}}\right)\cdot\left(\frac{n}{m}\right)^\frac{n}{2r^2}\right)^m\\
&\le e^{-\frac{1}{\sqrt{2\pi}}\left(\frac{r\sqrt{n\log{\frac{m}{n}}}}{r^2+n\log{\frac{m}{n}}}\right)\cdot\left(\frac{n}{m}\right)^\frac{n}{2r^2}\cdot m}.
\end{align*}

Therefore, by union bound, it follows that
\begin{align*}
P_r
&\le
e^{-\frac{1}{\sqrt{2\pi}}\left(\frac{r\sqrt{n\log{\frac{m}{n}}}}{r^2+n\log{\frac{m}{n}}}\right)\cdot\left(\frac{n}{m}\right)^\frac{n}{2r^2}\cdot
m+n\log{\frac{10r}{\sqrt{n}}}}\\
&\le e^{-n\log{\frac{10r}{\sqrt{n}}}}\le \left(\frac{\sqrt{n}}{10r}\right)^n.
\end{align*}



\noindent Case 2. Let $r>\sqrt{n\log{(m/n)}}$. Since $A_ix$ is distributed according to $N(0,r^2\sigma^2)$, by Lemma \ref{lemma:gaussian-anti-concentration-bound}, we have that
\[
\prob{A_ix\le \sigma\sqrt{n\log{\frac{m}{n}}}}\le \frac{1}{r}\sqrt{\frac{1}{2\pi} n\log{\frac{m}{n}}} \le \frac{4}{5r} \sqrt{n\log{\frac{m}{n}}}.
\]
The random variables $A_1x,\ldots,A_mx$ are independent and identically distributed. Therefore,
\begin{align*}
P_x
&= \prod_{i=1}^{m} \prob{|A_ix|\le \sigma\sqrt{n\log{\frac{m}{n}}}}\\
&\le \left(\frac{4}{5r} \sqrt{n\log{\frac{m}{n}}}\right)^m.
\end{align*}
Hence, by union bound,
\begin{align*}
P_r
&\le e^{-n\left(\frac{m}{n}\log{\left(\frac{5r}{4\sqrt{n\log{\frac{m}{n}}}}\right)}-\log{\frac{10r}{\sqrt{n}}}\right)}\\
&\le e^{-n\left(\frac{m}{2n}\log{\left(\frac{5r}{4\sqrt{n\log{\frac{m}{n}}}}\right)}\right)}\\ 
&\le \left(\frac{4\sqrt{n\log{\frac{m}{n}}}}{5r}\right)^{\frac{m}{2}}.
\end{align*}


Finally,
\begin{align*}
\prob{\exists x\in \cup_{r\ge \sqrt{n}}U_r:A_ix\le \sigma\sqrt{n\log{\frac{m}{n}}}\ \forall i\in [m]}
&= \sum_{r\ge \sqrt{n}}P_r
\end{align*}
\begin{align*}
\sum_{r\ge \sqrt{n}}P_r
&= \sum_{r\in \left[\sqrt{n},\sqrt{n\log{\frac{m}{n}}}\right]} P_r + \sum_{r>\sqrt{n\log{\frac{m}{n}}}} P_r\\
&\le \frac{1}{10^n} \int_{r=\sqrt{n}}^{\infty}\left(\frac{\sqrt{n}}{r}\right)^{n}dr
+ \left(\frac{4}{5}\right)^{\frac{m}{2}} \int_{r=\sqrt{n\log{\frac{m}{n}}}}^{\infty}\left(\frac{\sqrt{n\log{\frac{m}{n}}}}{r}\right)^{\frac{m}{2}}dr\\
&\le \frac{1}{10^n}\cdot\frac{\sqrt{n}}{n-1} + \left(\frac{4}{5}\right)^{\frac{m}{2}}\cdot\left(\frac{2\sqrt{n\log{\frac{m}{n}}}}{m-2}\right)\\
&\le \frac{1}{2^{n}} \quad \quad \text{(since $m\ge 1000n$).}\\
\end{align*}
\end{proof}

\begin{proof}[Proof of Lemma \ref{lemma:linear-discrepancy-lower-bound}]
There exists $x\in \integers^n$ such that
\[
A_i(x-x_0)\le \frac{\sigma}{2}\sqrt{n\log{\frac{m}{n}}}\ \forall i\in [m]
\]
if and only if there exists $x\in \integers^n \cap\{x\in \reals^n:\abs{x_j}\ge 1\ \forall\ j\in [n]\}$ such that
\[
A_ix\le \sigma\sqrt{n\log{\frac{m}{n}}}\ \forall i\in [m].
\]
The result follows by Lemma \ref{lemma:discrepancy-lower-bound1}.
\end{proof}

\section{Proof of Theorem \ref{thm:threshold-radius}}
We now have all the ingredients needed prove Theorem \ref{thm:threshold-radius}.
\begin{proof}[Proof of Theorem \ref{thm:threshold-radius}]
Let $P=\{x\in \reals^n:\vproduct{a_i}{x}\le b_i \ \forall\ i\in[m]\}$, where each $a_i$ is chosen from a spherically symmetric distribution. Then $\alpha_i=a_i/\norm{a_i}$ for $i\in [m]$ is distributed randomly on the unit sphere. A random unit vector $\alpha_i$ can be obtained by drawing each coordinate from the Gaussian distribution $N(0,\sigma^2=1/n)$ and normalizing the resulting vector. Thus, we may assume $\alpha_i=A_i/\norm{A_i}$ where each coordinate $A_{ij}$ is drawn from the Gaussian distribution $N(0,1/n)$. Here, we show that the probability that there exists a vector $A_i$ that gets scaled by more than a constant is at most $2me^{-n/96}$.

Taking $r=n$ and $\sigma^2=1/n$ in Lemma \ref{lemma:bounded-norm-gaussian-variables}, we have
\[
\prob{\exists i\in [m]:|\norm{A_i}^2-1|>\frac{1}{2}}\le 2me^{-\frac{n}{96}}.
\]
Hence, with probability at least $1-2me^{-n/96}$, we have that $\sqrt{1/2} \le \norm{A_i}\le \sqrt{3/2}$ for every $i\in [m]$. We now show the upper and lower bounds.
\begin{enumerate}
\item Since $P$ contains a ball of radius $\radius_1$, it follows that $P\supseteq Q$ where
\[
Q=\{x\in \reals^n|\abs{\vproduct{\alpha_i}{(x-x_0)}}\leq \radius_1 \text{ for } i\in \setm\}.
\]
Using Theorem \ref{thm:random-hyperplanes-arbitrary-center-contain-integer-point} and $\sigma^2=1/n$, we know that there exists a randomized algorithm that takes as input $A$ and $x_0$ and outputs an integer point $x\in \integers^n$ such that for every $i\in \setm$
\[
|\vproduct{A_i}{(x-x_0)}|\leq 480\left(\sqrt{\log{\frac{m}{n}}}+\sqrt{\frac{\log{m}\log{(mn)}}{n}\log{\frac{m}{\log{m}}}}\right).
\]
with probability at least $1-(4/m^3)$.
Thus, with probability at least $1-(4/m^3)-2me^{-n/96}$, we obtain $x\in \integers^n$ satisfying
\[
\abs{\vproduct{\alpha_i}{(x-x_0)}}
=\frac{|\vproduct{A_i}{(x-x_0)}|}{\norm{A_i}}
\leq 960\left(\sqrt{\log{\frac{m}{n}}}+\sqrt{\frac{\log{m}\log{(mn)}}{n}\log{\frac{m}{\log{m}}}}\right)
\]
for every $i\in \setm$. Thus we have an integer point in the polytope $Q$ and hence, an integer point in $P$.
\item For $x_0=(1/2,\ldots,1/2)$, let
\[
P=\left\{x\in \reals^n:\vproduct{A_i}{(x-x_0)}\le \norm{A_i}\sqrt{\frac{1}{6}\log{\frac{m}{n}}}\ \forall i\in [m]\right\}.
\]
Then, $P$ contains a ball of radius $R_0$ centered around $x_0$ and hence is an instance of the random polytope $P(n,m,x_0,R_0)$. Further, with probability at least $1-2me^{-n/96}$, $P$ is contained in
\[
Q=\left\{x\in \reals^n:\vproduct{A_i}{(x-x_0)}\le \frac{1}{2}\sqrt{\log{\frac{m}{n}}}\ \forall i\in [m]\right\}.
\]
By Lemma \ref{lemma:linear-discrepancy-lower-bound}, with probability at least $1-2^{-n}$, we have that $Q\cap \integers^n=\emptyset$. Thus, with probability at least $1-2^{-n}-2me^{-n/96}$, we have that $P\cap\integers^n=\emptyset$.
\end{enumerate}
\end{proof}

\section{Open Questions}
Propositions \ref{prop:linear-disc-IP} and \ref{prop:discrepancy-IP} hold for arbitrary constraint matrices describing the polytope. It would be interesting to understand if these observations could be used to solve the IP formulations of combinatorial feasibility/optimization problems. A concrete question is whether we can efficiently compute discrepancy/linear discrepancy for a reasonably general family of matrices.

Another open question is to perform optimization on random polytopes, even along random objective directions. In particular, how efficient are cutting-plane and branch-and-bound algorithms for optimization on random polytopes?

From a purely probabilistic perspective, our work also raises the question of whether there exists a sharp threshold $R^*$ on the radius so that the random polytope $P(n,m,x_0,R)$ is integer infeasible with high probability for some center $x_0$ if $R\le R^*$ and is integer feasible for all center $x_0$ with high probability if $R>R^*$.

\paragraph{Acknowledgements.} We thank Joel Spencer and Will Perkins for helpful discussions about this work.

\bibliographystyle{abbrv}
\bibliography{references}


\end{document}